\documentclass[11pt]{article}
\usepackage{latexsym,amsmath,amssymb,theorem}

\newcommand{\QCBZ}{\mathsf{QCB_0}}
\newcommand{\QCB}{\mathsf{QCB}}
\newcommand{\QTE}{\mathsf{QTE}}
\newcommand{\QTEZ}{\mathsf{QTE_0}}
\newcommand{\CBZ}{\mathsf{CB_0}}
\newcommand{\Pomega}{{P\hspace*{-1pt}\omega}}
\newcommand{\bfSig}{\mathbf{\Sigma}}
\newcommand{\bfPi}{\mathbf{\Pi}}
\newcommand{\bfGamma}{\mathbf{\Gamma}}
\newcommand{\bfDelta}{\mathbf{\Delta}}
\newcommand{\dom}{\mathit{dom}}
\newcommand{\rng}{\mathit{rng}}
\newcommand{\graph}{\mathit{graph}}
\newcommand{\EQ}{\mathit{EQ}}
\newcommand{\sfC}{\mathsf{C}}
\newcommand{\sfF}{\mathsf{F}}
\newcommand{\calN}{\mathcal{N}}
\newcommand{\Nk}[1]{\mathbb{N}\langle#1\rangle}
\newcommand{\Rk}[1]{\mathbb{R}\langle#1\rangle}
\newcommand{\Nalpha}{\Nk{\alpha}}
\newcommand{\Nbeta}{\Nk{\beta}}
\newcommand{\Nlambda}{\Nk{\lambda}}
\newcommand{\Ra}{\mathbb{R}\langle\alpha\rangle}

\newcommand{\IS}{\mathbb{S}}

\newcommand{\Tprod}{\mathop{\textstyle\prod}}

\newtheorem{lemma}{Lemma}[section]
\newtheorem{theorem}[lemma]{Theorem}
\newtheorem{corollary}[lemma]{Corollary}
\newtheorem{proposition}[lemma]{Proposition}

\theorembodyfont{\normalfont}

\newtheorem{remark}[lemma]{Remark}
\newtheorem{remarks}[lemma]{Remarks}
\newtheorem{definition}[lemma]{Definition}

\newenvironment{proof}{{\noindent\bf Proof.}}{\hspace*{\fill}$\Box$\par\bigskip}
\newenvironment{proof*}[1]{{\noindent\bf Proof} (#1){\bf.}}{\hspace*{\fill}$\Box$\par\bigskip}

\begin{document}

\title{Hyperprojective Hierarchy of qcb$_0$-Spaces}

\author{Matthias Schr\"oder
\thanks{Supported by the FWF project ``Definability and computability''.}
\\Kurt G\"odel Research Center, University of Vienna
\\Austria
\\ and
\\Victor Selivanov\thanks{Supported by the DFG Mercator professorship at the University of W\"urzburg, by the RFBR-FWF project ``Definability and computability'', and by the RFBR project 13-01-00015.}\\A.P. Ershov
Institute of Informatics
Systems SB RAS\\
Novosibirsk, Russia}
\date{}

\maketitle

\begin{abstract} 
We extend the Luzin hierarchy of qcb$_0$-spaces introduced in \cite{scs13} to all countable ordinals, 
obtaining in this way the hyperprojective hierarchy of qcb$_0$-spaces. 
We generalize all main results of \cite{scs13} to this larger hierarchy. 
In particular, we extend the Kleene-Kreisel continuous functionals of finite types to the continuous functionals
of countable types and relate them to the new hierarchy. 
We show that the category of hyperprojective qcb$_0$-spaces has much better closure properties 
than the category of projective qcb$_0$-space. 
As a result, there are natural examples of spaces that are hyperprojective but not projective.


\bigskip

\emph{\begin{small}
 This is an extended version of the conference paper \cite{scs14}.
\end{small}}

\end{abstract}

%
%

\section{Introduction}\label{in}

A basic notion of Computable Analysis  \cite{wei00} is
the notion of an {\em admissible representation} of a
topological space $X$. This is a partial continuous surjection
$\delta$ from the Baire space $\calN$ onto $X$
satisfying a certain universality property
(see Subsection~\ref{sub:admiss} for some more details).
Such a representation of $X$ usually induces a reasonable
computability theory on $X$, and the class of admissibly
represented spaces is wide enough to include most spaces of
interest for Analysis or Numerical Mathematics. As shown by the
first author \cite{sch:phd}, this class coincides with the
class of the so-called qcb$_0$-spaces, i.e. $T_0$-spaces which are
quotients of countably based spaces, and it forms a cartesian
closed category (with the continuous functions as morphisms).
Thus, among qcb$_0$-spaces one meets many important
function spaces including the continuous functionals of finite
types \cite{kl59,kr59} interesting for several branches of logic
and computability theory.

Along with the mentioned nice properties of qcb$_0$-spaces, this
class seems to be too broad to admit a deep understanding.
Hence, it makes sense to search for natural
subclasses of this class which still include ``practically''
important spaces but are (hopefully) easier to study. Interesting
examples of such subclasses are obtained if we consider, for each
level $\bfGamma$ of the classical Borel or Luzin (projective)
hierarchies of Descriptive Set Theory  \cite{ke95}, the
class of spaces which have an admissible representation of the
complexity $\bfGamma$ (below we make this precise). A study of the
resulting Borel and Luzin
hierarchies of qcb$_0$-spaces was undertaken in \cite{scs13}. 
In particular, it was shown that the Luzin hierarchy of qcb$_0$-spaces is closely related 
to the Kleene-Kreisel continuous functionals of finite types, 
and that the category of projective qcb$_0$-spaces is cartesian closed.

However, the class of  projective qcb$_0$-spaces is in a sense too restricted. 
In particular, it is not closed under some natural  constructions 
(e.g., countable products and coproducts) and does not contain some spaces of interest for Computable Analysis.

In this paper we extend the Luzin hierarchy of qcb$_0$-spaces  to all countable ordinals, 
obtaining in this way  the hyperprojective hierarchy of qcb$_0$-spaces.
We generalize to this larger hierarchy all main results of \cite{scs13} concerning the Luzin hierarchy. 
In particular, we extend the Kleene-Kreisel continuous functionals of finite types 
to the continuous functionals of countable types and relate them to the new hierarchy. 
We show that the category of hyperprojective qcb$_0$-spaces has much better closure properties 
than the category of projective qcb$_0$-space. 
As a result, there are natural examples of spaces that are hyperprojective but not projective.

After recalling some notions and known facts in the next section,
we summarize some basic  facts on the hyperprojective hierarchy of sets in Section \ref{sec:hph:sets}.
In Section \ref{sec:hph:qcb} we study the hyperprojective hierarchy of qcb$_0$-spaces,  in particular we show that 
the category of hyperprojective qcb$_0$-spaces is closed under countable limits and countable colimits
and function spaces. 
In Section \ref{sec:functional} we introduce and study the continuous functionals of countable types
and relate them in Section~\ref{sec:KKCF:HP} to the hyperprojective hierarchy of qcb$_0$-spaces. 
In Section \ref{categ} we establish some properties of  categories of hyperprojective qcb$_0$-spaces.
Finally, in Section \ref{sec:final} we provide additional natural examples 
of hyperprojective qcb$_0$-spaces.

%
%

\section{Notation and preliminaries}\label{prelim}

\subsection{Notation}\label{subnot}

We freely use the standard set-theoretic notation like
$\dom(f),\rng(f)$ and $\graph(f)$ for the domain, range and graph of
a function $f$, respectively, $X\times Y$ for the Cartesian
product, $X \oplus Y$ for the disjoint union of sets $X$
and $Y$, $Y^X$ for the set of functions $f \colon X\to Y$
(but in the case when $X,Y$ are qcb$_0$-spaces we use the same
notation to denote the set of continuous functions
from $X$ to $Y$), and $P(X)$ for the set of all subsets of $X$.
For $A\subseteq X$, $\overline{A}$ denotes the complement
$X\setminus A$ of $A$ in $X$. We identify the set of natural
numbers with the first infinite ordinal $\omega$. The first
uncountable ordinal is denoted by $\omega_1$.
The notation $f:X\to Y$
means that $f$ is a (total) function from a set $X$ to a set $Y$.

\subsection{Topological spaces}\label{sub:topspaces}

We assume the reader to be familiar with the basic notions of
topology. The collection of all open subsets of a topological space $X$
(i.e.\ the topology of $X$) is denoted by $\tau_X$;
for the underlying set of $X$ we will write $X$ in abuse of notation.
We will usually abbreviate ``topological space'' to ``space''.
Remember that a space is {\em zero-dimensional}, if it has a basis of clopen sets. 
A \emph{basis} for the topology on $X$ is a set $\cal B$ of open subsets of $X$  such that 
for every $x\in X$ and open $U$ containing $x$, there is $B\in \cal B$ satisfying $x\in B\subseteq U$. 
A space is \emph{countably based}, if it has a countable basis. 
By a cb$_0$-space we mean a countably based $T_0$-space. 
The class of $cb_0$-spaces is denoted by $\CBZ$.
We write $X \cong Y$, if $X$ and $Y$ are homeomorphic.

A space $Y$ is called a \emph{(continuous) retract} of a space $X$ if there
are continuous functions $s:Y\to X$ and $r:X\to Y$ such that
composition $rs$ coincides with the identity function $id_Y$ on
$Y$. Such a pair of functions $(s,r)$ is called a
{\em section-retraction} pair. Note that the section $s$ is a
homeomorphism between $Y$ and the subspace $s(Y)=\{x\in X\mid
sr(x)=x\}$ of $X$, and $s^{-1}=r|_{s(Y)}$.

Let $\omega$ be the space of non-negative integers with the
discrete topology. Of course, the spaces
$\omega\times\omega=\omega^2$, and $\omega\oplus\omega$ are
homeomorphic to $\omega$, the first homeomorphism is realized by
the Cantor pairing function $\langle \cdot,\cdot\rangle$.

Let ${\mathcal N}=\omega^\omega$ be the set of all infinite
sequences of natural numbers (i.e., of all functions $\xi \colon
\omega \to \omega$). Let $\omega^*$ be the set of finite sequences
of elements of $\omega$, including the empty sequence. For
$\sigma\in\omega^*$ and $\xi\in{\mathcal N}$, we write
$\sigma\sqsubseteq \xi$ to denote that $\sigma$ is an initial
segment of the sequence $\xi$. By $\sigma\xi=\sigma\cdot\xi$ we
denote the concatenation of $\sigma$ and $\xi$, and by
$\sigma\cdot\calN$ the set of all extensions of $\sigma$ in
$\calN$. For $x\in\calN$, we can write
$x=x(0)x(1)\dotsc$ where $x(i)\in\omega$ for each $i\in \omega$. For
$x\in\calN$ and $n\in \omega$, let $x^{<n}=x(0)\dotsc x(n-1)$
denote the initial segment of $x$ of length $n$. 

By endowing $\calN$ with the product of the discrete
topologies on $\omega$, we obtain the so-called \emph{Baire space}.
The product topology coincides with the topology
generated by the collection of sets of the form
$\sigma\cdot\calN$ for $\sigma\in\omega^*$. The Baire space
is of primary importance for Descriptive Set Theory  and Computable Analysis. The importance stems from
the fact that many countable objects are coded straightforwardly
by elements of $\calN$, and it has very specific topological properties. 
In particular, it is a perfect zero-dimensional space,
and the spaces ${\mathcal N}^2$, ${\mathcal N}^\omega$,
$\omega\times \calN=\calN \oplus \calN \oplus\dotsc$ (endowed with the product topology) 
are all homeomorphic to $\calN$. 
Let $(x,y)\mapsto\langle x,y\rangle$ be a homeomorphism between $\calN^2$ and $\calN$. 
Let $(x_0,x_1,\dotsc)\mapsto\langle x_0,x_1,\dotsc\rangle$ be the homeomorphism between 
$\calN^\omega$ and $\calN$ defined by $\langle x_0,x_1,\ldots\rangle\langle m,n\rangle=x_m(n)$.

The space $\Pomega$ is formed by the set of subsets of $\omega$ equipped with the Scott topology,
the basic open sets of which are the sets $\{A\subseteq\omega\mid F\subseteq A\}$, 
where $F$ ranges over the finite subsets of $\omega$.
It has the following well-known universality property:

\begin{proposition}\label{inj2}
 A topological space $X$ embeds into $\Pomega$ iff $X$ is a $cb_0$-space.
\end{proposition}

Remember that a space $ X $ is \emph{Polish}, if it is countably based
and metrizable with a metric $d$ such that $(X,d)$
is a complete metric space. Important examples of Polish spaces are $\omega$, $\mathcal{N}$,  
the space of reals $ \mathbb{R} $ and its Cartesian powers $ \mathbb{R}^n $ ($ n < \omega $),
the closed unit interval $ [0,1] $,
the Hilbert cube $ [0,1]^\omega $ and the Hilbert space $ \mathbb{R}^\omega $. 
Simple examples of  non-Polish spaces are the \emph{Sierpinski space} $\IS=\{\bot,\top\}$,
where the set $\{\top\}$ is open but not closed, and the space  of rationals.


\subsection{Admissible representations and qcb$_0$-spaces}
\label{sub:admiss}

A {\em representation} of a space $X$ is a
surjection of a subspace of the Baire space $\calN$ onto $X$.
A representation $\delta$ of $X$ is {\em admissible}, if
it is continuous and
any continuous function $\nu:Z \to X$ from a subspace
$Z$ of $\calN$ to $X$ is continuously reducible to $\delta$,
i.e. $\nu=\delta g$ for some continuous function $g:Z \to \calN$.
A topological space is {\em admissibly representable} if it has
an admissible representation.

The notion of admissibility was introduced in \cite{kw85}
for representations of countably based spaces
(in a different but equivalent formulation)
and was extensively studied by many authors.
In \cite{sch:ext,sch:phd} the notion was extended to non-countably based spaces
and a nice characterization of the admissibly represented spaces was achieved. 
Namely, the admissibly represented sequential topological spaces
coincide with the qcb$_0$-spaces.
Spaces which arise as topological quotients of countably based spaces are called \emph{qcb-spaces}, 
and qcb-spaces that have the $T_0$-property are called qcb$_0$-spaces.

The category $\QCB$ of qcb-spaces as objects and continuous functions 
as morphisms is known to be cartesian closed (cf.\ \cite{ELS04,sch:phd}).
The same is true for its full subcategory $\QCBZ$ of qcb$_0$-spaces.
The exponential $Y^X$ to qcb-spaces $X,Y$ has
the set of continuous functions from $X$ to $Y$ as the underlying set,
and its topology is the sequentialization of the compact-open topology on $Y^X$.
By the \emph{sequentialization} of a topology $\tau$ we mean
the family of all sequentially open sets pertaining to this topology.
(Remember that \emph{sequentially open} sets are defined to be the complements
of the sets that are closed under forming limits of converging sequences.)
The sequentialization of $\tau$ is finer than or equal to $\tau$.
The topology of the $\QCB$-product to $X$ and $Y$, which we denote by $X \times Y$,
is the sequentialization of the well-known Tychonoff topology
on the cartesian product of the underlying sets of $X$ and $Y$.
If $X$ and $Y$ are additionally $T_0$, then $X \times Y$ and $Y^X$ are $T_0$ as well.
So products and exponentials in $\QCB$ and in $\QCBZ$ are formed in the same way
as in its supercategory $\mathsf{Seq}$ of sequential topological spaces.

We will also use the following well-known facts (see e.g.\ \cite{sch:phd,wei00}).

\begin{proposition}\label{p:univ}
There is a partial continuous function $u:\subseteq\calN^2\to\calN$ such that
$\dom(u)\in\bfPi^0_2(\calN^2)$, 
for any partial continuous function $g$ on $\calN$ there is some $p\in\calN$ 
such that $u_p:=\lambda x.u(p,x)$ is an extension of $g$,
and for any partial continuous function $G:\subseteq \calN \times \calN \to \calN$ 
there is a total continuous function $g$ on $\calN$ such that $u(g(p),q)=F(p,q)$ for all $(p,q) \in \dom(G)$. 
\end{proposition}

We describe the construction of canonical admissible representations
for products and function spaces formed in $\QCBZ$ (cf.\ \cite{sch:phd}).

\begin{proposition}\label{p:fspace:rep}
 Let $\delta$ and $\gamma$ be admissible representations
 for qcb$_0$-spaces $X$ and $Y$, respectively.
 Then canonical admissible representations $[\delta\times\gamma]$ for the $\QCBZ$-product $X \times Y$
 and $[\gamma^\delta]$ for the $\QCBZ$-exponential $Y^X$ can be defined by:
 \[
   \text{ $[\delta\times\gamma](\langle p,q \rangle)=(x,y)$
          iff $\delta(p)=x \wedge \gamma(q)=y$}
   \quad\mathrm{and}\quad
   \text{$[\gamma^\delta](p)=f$ iff $f\delta=\gamma u_p|_{\dom(\delta)}$.}
 \]
 for $p,q \in \calN$, $x \in X$, $y \in Y$,
 and $f\colon X \to Y$.
\end{proposition}

Along with the binary product, the category $\QCBZ$ is  closed under countable product. The countable product of a sequence $X_0,X_1,\ldots$ of qcb$_0$-spaces is denoted by $\prod_{n\in \omega}X_n$. If $\delta_n$ is an admissible representation of $X_n$ for each $n$, the induced admissible representation of $\prod_{n\in \omega}X_n$ is denoted by $\prod_{n\in \omega}\delta_n$.  The category $\QCBZ$ is also closed under  other constructions including countable limits and colimits.
We will discuss some such constructions in Section~\ref{sec:hph:qcb} where we also recall some relevant notions of Category Theory.

For any qcb$_0$-space $X$, 
let $\mathcal{O}(X)$ be the hyperspace of open sets in $X$ endowed with the Scott topology.
The space $\mathcal{O}(X)$ is well known to be homeomorphic to the function space $\mathbb{S}^X$,
where $\mathbb{S}$ is the Sierpinski space.

%
%

\section{Hyperprojective hierarchy of sets}\label{sec:hph:sets}

Here we recall some facts on hierarchies in arbitrary spaces, with the emphasis on the hyperprojective hierarchy 
in the Baire space. Additional information on the hyperprojective hierarchy may be found in \cite{ke83}.

A {\em pointclass} on $ X $ is simply a collection $
\bfGamma(X) $ of subsets of $ X $. A
{\em family of pointclasses} \cite{s13} is a family $
\bfGamma=\{\bfGamma(X)\} $ indexed by arbitrary
topological spaces $X$ such that each $ \bfGamma(X) $ is
a pointclass on $ X $ and $ \bfGamma $ is closed under
continuous preimages, i.e. $ f^{-1}(A)\in\bfGamma(X) $
for every $ A\in\bfGamma(Y) $ and every continuous
function $ f \colon X\to Y $. A basic example of a
family of pointclasses is given by the family
$\mathcal{O}=\{\tau_X\}$ of the topologies of all the spaces $X$.

We will use some operations on families of pointclasses. First, the usual
set-theoretic operations will be applied to the families of
pointclasses pointwise: for example, the union $\bigcup_i
\bfGamma_i$ of the families of pointclasses
$\bfGamma_0,\bfGamma_1,\ldots$ is defined by
$(\bigcup_i\bfGamma_i)(X)=\bigcup_i\bfGamma_i(X)$.

Second, a large class of such operations is induced by the
set-theoretic operations of L.V. Kantorovich and E.M. Livenson
(see e.g. \cite{s13} for the general definition). 
Among them are the operation $\bfGamma\mapsto\bfGamma_\sigma$,
where $\bfGamma(X)_\sigma$ is the set of all countable unions of sets in $\bfGamma(X)$, 
the operation $\bfGamma\mapsto\bfGamma_\delta$,
where $\bfGamma(X)_\delta$ is the set of all countable intersections of sets in $\bfGamma(X)$, 
the operation $\bfGamma\mapsto\bfGamma_c$,
where $\bfGamma(X)_c$ is the set of all complements of sets in $\bfGamma(X)$, 
the operation $\bfGamma\mapsto\bfGamma_d$,
where $\bfGamma(X)_d$ is the set of all differences of sets in $\bfGamma(X)$, 
the operation $\bfGamma\mapsto\bfGamma_\exists$ defined by 
$\bfGamma_\exists(X):=\{\exists^\mathcal{N}(A)\mid A\in\bfGamma(\calN\times X)\}$,
where $\exists^\mathcal{N}(A):=\{x\in X\mid \exists p\in\calN.(p,x)\in A\}$ is the projection of $A\subseteq\calN\times X$ 
along the axis $\calN$,
and finally the operation $\bfGamma\mapsto\bfGamma_\forall$ defined by 
$\bfGamma_\forall(X):=\{ \forall^\mathcal{N}(A) \mid A\in\bfGamma(\calN\times X)\}$,
where $\forall^\mathcal{N}(A):=\{x\in X\mid \forall p\in\calN.(p,x)\in A\}$.

The  operations on families of pointclasses enable to provide short uniform descriptions 
of the classical hierarchies in arbitrary spaces. 
E.g., the Borel  hierarchy is the family of pointclasses
$\{\bfSig^0_\alpha\}_{ \alpha<\omega_1}$ defined by induction on $\alpha$ as follows \cite{s06,br}:
 $$\bfSig^0_0(X):=\{\emptyset\}, \bfSig^0_1 := \mathcal{O}, \bfSig^0_2 := (\bfSig^0_1)_{d\sigma},\text{ and } \bfSig^0_\alpha(X) :=(\bigcup_{\beta<\alpha}\bfSig^0_\beta(X))_{c\sigma}$$
 for $\alpha>2$.
The sequence $\{\bfSig^0_\alpha(X)\}_{ \alpha<\omega_1}$ is called \emph{the Borel hierarchy} in $X$.
We also let $\bfPi^0_\beta(X) := (\bfSig^0_\beta(X))_c $
and $\bfDelta^0_\alpha(X) := \bfSig^0_\alpha(X) \cap \bfPi^0_\alpha (X)$.
The classes $\bfSig^0_\alpha(X),\bfPi^0_\alpha(X),{\bf\Delta}^0_\alpha(X)$
are called the \emph{levels} of the Borel hierarchy in $X$.

For this paper, the hyperprojective hierarchy is of main interest. 

\begin{definition}\label{hpro}
 The hyperprojective hierarchy is the family of pointclasses
 $\{\bfSig^1_\alpha\}_{ \alpha<\omega_1}$ defined by induction on $\alpha$ as follows:  
 $\bfSig^1_0=\bfSig^0_2$, 
 $\bfSig^1_{\alpha+1}=(\bfSig^1_\alpha)_{c\exists}$, 
 $\bfSig^1_\lambda=(\bfSig^1_{<\lambda})_{\delta \exists}$, 
 where $\alpha,\lambda<\omega_1$, $\lambda$ is a limit ordinal, 
 and $\bfSig^1_{<\lambda}(X):=\bigcup_{\alpha<\lambda}\bfSig^1_\alpha(X)$.
\end{definition}

 In this way, we obtain for any topological space $X$ the sequence 
 $\{\bfSig^1_\alpha(X)\}_{\alpha<\omega_1}$, which we call here \emph{the hyperprojective hierarchy in $X$}. 
 The pointclasses $\bfSig^1_\alpha(X)$, $\bfPi^1_\alpha(X):=(\bfSig^1_\alpha(X))_c$ and 
 $\bfDelta^1_\alpha(X):=\bfSig^1_\alpha(X)\cap\bfPi^1_\alpha(X)$ are called 
 \emph{levels of the hyperprojective hierarchy in $X$}. 
 The finite non-zero levels of the hyperprojective hierarchy coincide with the corresponding levels 
 of the Luzin's projective hierarchy \cite{br,scs13}. 
 The class of \emph{hyperprojective sets} in $X$ is defined as the union of all levels 
 of the hyperprojective hierarchy in $X$. 
 
\begin{remarks}\label{p:rem} \rm
\begin{enumerate}
\item
 If $X$ is Polish then one can equivalently take $\bfSig^1_0=\bfSig^0_1$ 
 in the definition of the hyperprojective hierarchy and obtain the same non-zero levels  as above. For non-Polish spaces our definition guarantees the ``right'' inclusions of the levels,  as the first item of the next proposition states. 

\item 
 In the case of Polish spaces our ``hyperprojective hierarchy'' is in fact an initial segment of the hyperprojective hierarchy  from \cite{ke83}, so ``$\omega_1$-hyperprojective'' would be more precise name for our hierarchy; nevertheless, we prefer to use the easier term ``hyperprojective'' for our hierarchy. 
  
\item In the literature one can find two slightly different definitions of hyperprojective hierarchy. Our definition corresponds to that in \cite{ke83}. The other one (see e.g. Exercise 39.18 in \cite{ke95}) differs from ours only for limit levels, namely it takes $(\bfSig^1_{<\lambda})_\sigma$ instead of our $\bfSig^1_\lambda$. Our choice simplifies several formulations in Sections \ref{sec:functional} and \ref{categ}. The classes $(\bfSig^1_{<\lambda})_\sigma$ and their duals are also of interest for this paper, they are used in Sections \ref{sec:hph:qcb} and \ref{sec:functional}.

\end{enumerate}
\end{remarks}

 The next assertion collects  some   properties of the hyperprojective hierarchy. 
 They are proved just in the same way as for the classical projective hierarchy in Polish spaces \cite{ke95}, we omit the corresponding details.
 
\begin{proposition}\label{p:hpro1}
\begin{enumerate}
\item
 For any $\alpha < \beta<\omega_1$, 
 $\bfSig^1_\alpha\cup\bfPi^1_\alpha\subseteq\bfDelta^1_\beta$.
\item
  For any limit countable ordinal $\lambda$, $\bfSig^1_{<\lambda}=\bfPi^1_{<\lambda}$ and $(\bfSig^1_{<\lambda})_\delta=(\bfPi^1_{<\lambda})_{\sigma c}$.
\item
 For any non-zero $\alpha<\omega_1$, 
 $\bfSig^1_\alpha=(\bfSig^1_\alpha)_\sigma=(\bfSig^1_\alpha)_\delta=(\bfSig^1_\alpha)_\exists$. 
 In particular, the class $\bfSig^1_\alpha(\mathcal{N})$ is closed under countable unions, 
 countable intersections, continuous images, 
 and continuous preimages of functions with a $\bfPi^0_2$-domain. 
\item
 For any non-zero $\alpha<\omega_1$, 
 $\bfPi^1_\alpha=(\bfPi^1_\alpha)_\sigma=(\bfPi^1_\alpha)_\delta=(\bfPi^1_\alpha)_\forall$.  
 In particular, the class $\bfPi^1_\alpha(\mathcal{N})$ is closed under countable unions and countable intersections,
 and continuous preimages of functions with a $\bfPi^0_2$-domain.
 \item
  For any limit countable ordinal $\lambda$, the family of pintclasses $(\bfSig^1_{<\lambda})_\sigma$ is closed under continuous preimages and the operations $\sigma$ of countable union,  $\exists$ of projection along $\mathcal{N}$ and finite intersection but not under the operation $\delta$ of countable intersection. 
\item
 For any uncountable Polish space (and also for any uncountable  quasi-Polish space \cite{br}) $X$, 
 the hyperprojective hierarchy in $X$ does not collapse, i.e. $\bfSig^1_\alpha(X)\not\subseteq\bfPi^1_\alpha(X)$ 
 for each $\alpha<\omega_1$.
\end{enumerate}
\end{proposition}

%
%

\section{Hyperprojective hierarchy of qcb$_0$-spaces}\label{sec:hph:qcb}

Here we discuss the hyperprojective hierarchy of qcb$_0$-spaces.
In particular we extend all results from \cite{scs13} concerning the Luzin's projective hierarchy of qcb$_0$-spaces.

For any representation $\delta$ of a space $X$, let
$\mathit{EQ}(\delta):=\{( p,q) \in \calN^2 \mid p,q\in \dom(\delta)\wedge\delta(p)=\delta(q)\}$.
Let $\bfGamma$ be a family of pointclasses. 
A topological space $X$ is called \emph{$\bfGamma$-representable}, if $X$ has an
admissible representation $\delta$ with $\mathit{EQ}(\delta)\in\bfGamma(\calN^2)$. 
The class of all $\bfGamma$-representable spaces is denoted $\QCBZ(\bfGamma)$. 
This notion from \cite{scs13} enables to transfer hierarchies of sets to the corresponding hierarchies 
of qcb$_0$-spaces. In particular, we arrive at the following definition.

\begin{definition}\label{def:hierqcb}
  The sequence
  $\{\QCBZ(\bfSig^1_\alpha)\}_{\alpha<\omega_1}$ is called the \emph{hyperprojective hierarchy} of qcb$_0$-spaces.
  By \emph{levels} of this hierarchy we mean the classes $\QCBZ(\bfSig^1_\alpha)$ as well
  as the classes $\QCBZ(\bfPi^1_\alpha)$ and $\QCBZ(\bfDelta^1_\alpha)$.
\end{definition}

The next assertion summarizes extensions of the corresponding results from \cite{scs13} about the Luzin hierarchy. 
They are proved just in the same way as in \cite{scs13}.

\begin{proposition}\label{p:equiv:Gamma-representable}
\begin{enumerate}
\item  
 Let $\bfGamma\in\{\bfSig^1_\alpha, \bfPi^1_\alpha \mid 0\leq\alpha<\omega_1\}$
 and let $X$ be a Hausdorff space.
 Then $X$ is $\bfGamma$-representable,
 if $X$ has an admissible representation $\delta$ with $\dom(\delta)\in \bfGamma(\calN)$.
\item
 Let $\bfGamma\in\{\bfSig^1_\alpha, \bfPi^1_\alpha\mid 0\leq\alpha<\omega_1\}$.
 Then any continuous retract of a
 $\bfGamma$-representable space is a $\bfGamma$-representable space.
\item
 The hyperprojective hierarchy  of qcb$_0$-spaces does not collapse, more precisely,
 $\QCBZ(\bfSig^1_\alpha)\not\subseteq \QCBZ(\bfPi^1_\alpha)$ for each  $\alpha<\omega_1$.
\item
 For any 
 $\bfGamma \in \{\bfSig^1_\alpha, \bfPi^1_\alpha \mid 1\leq\alpha<\omega_1\}$, we have
 $\QCBZ(\bfGamma) \cap \CBZ = \CBZ(\bfGamma)$,
 where $\CBZ(\bfGamma)$ is the class of spaces homeomorphic to a $\bfGamma$-subspace of $\Pomega$.
\end{enumerate}
\end{proposition}

Now we establish some closure properties of the hyperprojective hierarchy of qcb$_0$-spaces.
We begin with exponentiation in $\QCBZ$.
The next proposition extends and improves Theorem 7.1 in \cite{scs13}.

\begin{proposition}\label{p:exp} 
 Let $1\leq \alpha< \omega_1$,
 $X\in \QCBZ(\bfSig^1_\alpha)$ and $Y\in \QCBZ(\bfPi^1_\alpha)$.
 Then $Y^X\in \QCBZ(\bfPi^1_{\alpha})$.
\end{proposition}

\begin{proof}
 Let $\delta$ be a quotient representation of a qcb-space $X$ with 
 $\EQ(\delta) \in \bfSig^1_\alpha(\calN^2)$, and let $\gamma$ be an admissible representation 
 of a qcb$_0$-spce $Y$ with $\EQ(\gamma) \in \bfPi^1_\alpha(\calN^2)$.
 By Proposition~4.2.5 in \cite{sch:phd}, $\gamma^\delta$ is an admissible representation of $Y^X$.
 Using the universal function $u$ from Section~\ref{sub:admiss},
 we define the subset $M \subseteq \calN^4$ by
 \begin{align*}
  M:=\Big\{ & (p_1,p_2,q_1,q_2) \in \calN^4 \,\Big| \, 
     (p_1,q_1),(p_1,q_2),(p_2,q_1),(p_2,q_2) \in \dom(u) \;\;\text{and}
   \\
   & \hspace*{4em}
    (u(p_1,q_1),u(p_1,q_2)), (u(p_1,q_1),u(p_2,q_1)), (u(p_2,q_1),u(p_2,q_2)) \in \EQ(\gamma)
  \Big\}. 
 \end{align*}
 By the definition of $\gamma^\delta$ we have for all $p_1,p_2 \in \calN$ 
 \begin{align*}
  (p_1,p_2) \in \EQ(\gamma^\delta)
   &\;\Longleftrightarrow\; 
   \forall (q_1,q_2) \in \EQ(\delta).\, (p_1,p_2,q_1,q_2) \in M 
 \\
   &\;\Longleftrightarrow\; 
   \forall (q_1,q_2) \in \calN^2.\big( (q_1,q_2) \notin \EQ(\delta) \vee (p_1,p_2,q_1,q_2) \in M \big). 
 \end{align*}
 Since $\bfPi^1_\alpha$ is closed under forming preimages of partial continuous functions
 with a $\bfPi^0_2$-domain, $M$ is in $\bfPi ^1_\alpha(\calN^4)$.
 Thus $\EQ(\gamma^\delta) \in \bfPi^1_\alpha(\calN^2)$, because 
 $\bfPi^1_\alpha=(\bfPi^1_\alpha)_\forall=(\bfPi^1_\alpha)_\sigma$ (see Proposition~\ref{p:hpro1}).
 Hence $Y^X\in \QCBZ(\bfPi^1_{\alpha})$.
\end{proof}

The next proposition provides some complexity bounds on products and co-products
formed in $\QCBZ$.

\begin{proposition}\label{p:prod:coprod}
\begin{enumerate}
\item
  Any non-zero level of the hyperprojective hierarchy of qcb$_0$-spaces 
  is closed under countable $\QCBZ$-products and coproducts.
\item
  Let $\lambda$ be a countable limit ordinal. Let $\{X_k\}_{k\in \omega}$ be a sequence of qcb$_0$-spaces 
  such that $X_k\in\QCBZ(\bfSig^1_{<\lambda})$ for all $k$.
 Then the $\QCBZ$-product $\prod_{k\in \omega}X_k$ is in $\QCBZ((\bfSig^1_{<\lambda})_\delta)$
 and
 the co-product $\bigoplus_{k\in \omega}X_k$ is in $\QCBZ((\bfSig^1_{<\lambda})_\sigma)$. 
\end{enumerate}
\end{proposition}

\begin{proof} 
  Let $\gamma_k \colon D_k\to X_k$ be an admissible representation of $X_k$ for every $k \in \omega$.
 By \cite{sch:phd}, admissible representations
 $\gamma=\prod_k\gamma_k \colon D \to \prod_{k \in \omega} X_k$ and $\phi \colon E \to \bigoplus_{k \in \omega} X_k$
 for the product and the co-product can be defined by
 \[
   \gamma(p):=\big( \gamma_0(\pi_0(p)), \gamma_1(\pi_1(p)), \gamma_2(\pi_2(p)),\dotsc \big)
   \quad\text{and}\quad
   \phi(q) := \big( q(0),\gamma_{q(0)}(t(q)) \big)
 \]
 for all $p \in D:=\bigcap_{k \in \omega} \{ p \in \calN \,|\, \pi_k(p) \in D_k\}$
 and $q \in E:= \bigcup_{k \in \omega} \{q \in \calN \,|\, q(0)=k,\, t(q) \in D_k\}$.
 Here $\pi_k$ denotes the continuous $k$-th projection of the inverse of the
 homeomorphism $\langle . \rangle: \calN^\omega \to \calN$ from Section~\ref{sub:topspaces},
 and $t(q)$ denotes the sequence $q(1)q(2)q(3)\dotsc \in \calN$.
 We obtain:
 \begin{align*}
   \EQ(\gamma) &=\bigcap_{k \in \omega} (\pi_k \times \pi_k)^{-1}[\EQ(\gamma_k)] \,,
  \\
   \EQ(\phi) &=\bigcup_{k \in \omega}
    \big( (t \times t)^{-1} [\EQ(\gamma_k)] \cap \{ (p,q) \in \calN^2 \,|\, p(0)=q(0)=k \} \big) \,.
 \end{align*}
 The claims follow from these equations and Proposition~\ref{p:hpro1}.
\end{proof}

 For equalizers we have the following result.

\begin{proposition}\label{l:equalizers}  
 Let $\alpha$ be a non-zero countable ordinal.
 Then $\QCBZ(\bfSig^1_\alpha)$ and $\QCBZ(\bfPi^1_\alpha)$ are closed under forming equalizers.
\end{proposition}

\begin{proof}
  Let $\bfGamma \in \{ \bfSig^1_\alpha, \bfPi^1_\alpha \}$.
  Let $A,X \in \QCBZ(\bfGamma)$ and let $f_1,f_2 \colon X \to A$ be continuous.
  Then the sequential subspace $Y$ of $X$ with underlying set
  $\{ x \in X \,|\, f_1(x)=f_2(x)\}$ together with the inclusion map $\iota\colon Y \to X$
  form an equalizer for $f_1,f_2$ in $\QCBZ$. 
  Choose admissible representations $\phi$ and $\delta$ for $A$ and $X$, respectively,
  such that $\EQ(\phi),\EQ(\delta) \in \Gamma(\calN^2)$.
  Then the co-restriction $\gamma$ of $\delta$ to the underlying set of $Y$
  is known to be an admissible representation for $Y$.
  There are partial continuous functions $g_1,g_2$ on $\calN$ with $\bfPi^0_2$-domains
  which realize $f_1$ and $f_2$, respectively.
  Then
  $ \dom(\gamma)= \big\{ p \in  \dom(\delta) \,\big|\,(g_1(p),g_2(p)) \in \EQ(\phi) \big\}$.
  Since $g_1,g_2$ are continuous and have a $\bfPi^0_2$-domain, 
  $\dom(\gamma) \in \bfGamma(\calN)$ by Proposition~\ref{p:hpro1}.
  The equation $\EQ(\gamma)=(\dom(\gamma) \times \dom(\gamma)) \cap \EQ(\delta)$
  yields us $\EQ(\gamma) \in \bfGamma(\calN^2)$ and $Y \in \QCBZ(\bfGamma)$.
\end{proof}

Now we turn our attention to co-equalizers.
Co-equalizers in $\QCBZ$ are constructed by first forming a co-equalizer in the category $\QCB$
and then, if the resulting space is non-$T_0$, identifying points with the same neighbourhoods.
{Non-$T_0$} qcb-spaces do not have an admissible representation, but some of them have 
a quotient representation.
This motivates the following definition generalizing the one from above.
For a given family $\bfGamma$ of pointclasses, we say that a topological space $X$ is
\emph{$\bfGamma$-quotient-representable}, if $X$ has a quotient representation $\delta$
such that $\EQ(\delta) \in \bfGamma(\calN^2)$.
We denote the class of $\bfGamma$-quotient-representable spaces by $\QTE(\bfGamma)$
and the class of $\bfGamma$-quotient-representable $T_0$-spaces by $\QTEZ(\bfGamma)$.
Since any admissible representation of a sequential space is a quotient representation,
we have $\QCBZ(\bfGamma) \subseteq \QTEZ(\bfGamma)$.

We study the (non-uniform) descriptive complexity of the Kolmogorov operator $\mathcal{T}_0$
that maps any $T_0$-space to itself and sends a non-$T_0$-space $X$ to the quotient space
induced by the equivalence relation $\equiv_X$ given by the specification order of $X$,
i.e., $x \equiv_X x'$ iff $x$ and $x'$ have the same open neighbourhoods.

\begin{proposition}\label{p:QTE->QCBZ}  
 Let $\alpha$ be a non-zero countable ordinal.
 Then $X \in \QTE(\bfSig^1_\alpha)$ implies $\mathcal{T}_0(X) \in \QCBZ(\bfSig^1_{\alpha+2})$.
 Moreover, 
 $\QCBZ(\bfSig^1_\alpha) \subseteq \QTEZ(\bfSig^1_\alpha) \subseteq \QCBZ(\bfSig^1_{\alpha+2})$. 
\end{proposition}

The proof is based on the following slight improvement of Proposition~\ref{p:exp}.

\begin{proposition}\label{p:exp:QTE} 
 Let $1\leq \alpha< \omega_1$,
 $X\in \QTE(\bfSig^1_\alpha)$ and $Y\in \QCBZ(\bfPi^1_\alpha)$.
 Then $Y^X\in \QCBZ(\bfPi^1_{\alpha})$.
\end{proposition}

\begin{proof}
 The proof of Proposition~\ref{p:exp} has been formulated to show this.
\end{proof}

\begin{proof} \emph{(Proposition~\ref{p:QTE->QCBZ})}
 The Sierpinski space $\IS$ has $\{\bot,\top\}$ as its underlying set
 and $\{\emptyset,\{\top\},\linebreak[3]\{\bot,\top\}\}$ as its topology.
 The space $\mathcal{T}_0(X)$ is homeomorphic to the sequential subspace 
 of $\IS^{\IS^X}$ which has $\{ e(y) \,|\, y \in \mathcal{T}_0(X)\}$ 
 as its underlying set, cf.\ \cite{sch:phd}. 
 Here the continuous injection $e\colon \mathcal{T}_0(X) \to \IS^{\IS^X}$ 
 is defined by $e([x]_{\equiv_X})(h):=h(x)$.
 \par
 An admissible representation $\varrho_\IS$ for $\IS$
 is defined by $\varrho_\IS(p)=\bot \iff \forall i \in \omega. p(i)=0$.
 Clearly, $\EQ(\varrho_\IS)$ is a Boolean combination of open sets and thus in 
 $\bfPi^1_0(\calN^2) \subseteq\bfPi^1_\alpha(\calN^2)$.
 Two applications of Proposition~\ref{p:exp:QTE} yield that
 $\gamma:=\varrho_\IS^{[\varrho_\IS^\delta]}$ is an admissible representation 
 for $\IS^{\IS^X}$ with $\EQ(\gamma) \in \bfPi^1_{\alpha+1}(\calN^2)$ where $\delta$ is a quotient representation of $X$ such that $EQ(\delta)\in\mathbf\Sigma^1_\alpha(\mathcal{N}^2)$.
 We define a representation $\delta^*$ of $\mathcal{T}_0(X)$ 
 by $\delta^*(p)=y :\Longleftrightarrow \gamma(p)=e(y)$.
 By Proposition 4.3.2 in \cite{sch:phd}, $\delta^*$ is an admissible representation of $\mathcal{T}_0(X)$.
 Proposition~\ref{p:univ} yields a total continuous function $g\colon \calN \to \calN$
 with $u(g(r),q)=u(q,r)$ for all $(q,r) \in \dom(u)$. 
 One easily verifies $\gamma(g(r))=e([\delta(r)]_{\equiv_X})$ for all $r \in \dom(\delta)$.
 Hence for all $r \in \calN$ we have
 \[
   p \in \dom(\delta^*)
   \Longleftrightarrow
   \exists r \in \calN.\, \big( r \in \dom(\delta) \;\&\; (p,g(r)) \in \EQ(\gamma) \big) \,.
 \]
 By Proposition~\ref{p:hpro1}, $\dom(\delta^*) \in \bfSig^1_{\alpha+2}(\calN^2)$.
 Since $\EQ(\delta^*)=(\dom(\delta^*) \times \dom(\delta^*)) \cap \EQ(\gamma)$, 
 we obtain $\EQ(\delta^*) \in \bfSig^1_{\alpha+2}(\calN^2)$ 
 and $\mathcal{T}_0(X) \in \QCBZ(\bfSig^1_{\alpha+2})$.
 \\
 If $X$ is additionally a $T_0$-space, then $\mathcal{T}_0(X)=X$
 and thus $X \in \QCBZ(\bfSig^1_{\alpha+2})$.
\end{proof}

Now we can formulate our result about forming co-equalizers.

\begin{proposition}\label{l:QCBZ:coequalizers}
 Let $\lambda$ be a countable limit ordinal.
 Then $\QCBZ(\bfSig^1_{<\lambda})$ is closed under forming co-equalizers in $\QCBZ$.
\end{proposition}

\begin{proof}
 Let $\alpha<\lambda$.
 Let $A,X \in \QCBZ(\bfSig^1_\alpha)$ and 
 let $f,g \colon A \to X$ be continuous.
 A co-equalizer $q\colon X \to Y$ for $f,g$ in $\QCB$ is constructed as follows 
 (cf.\ \cite{sch:phd}):
 Let $\equiv$ be the equivalence relation obtained by the transitive closure of the relation
 $$R:=\{ (f(a),g(a)), (g(a),f(a)), (x,x) \,|\, a\in A,\, x \in X \}.$$
 Let $Y$ be the space that has the equivalence classes of $\equiv$ as its underlying set.
 The function $q$ is the surjection mapping $x \in X$ to its equivalence class $[x]_\equiv$.
 The topology of $Y$ is the quotient topology induced by $q$.
 If $Y$ is $T_0$, then $q$ is a co-equalizer for $f,g$ in $\QCBZ$ as well, otherwise
 the map $[.]_{\equiv_Y} \circ q\colon X \to \mathcal{T}_0(Y)$ yields a co-equalizer
 for $f,g$ in $\QCBZ$.
It remains to show that $\mathcal{T}_0(Y) \in \QCBZ(\bfSig^1_{<\lambda})$. 

 There are admissible representations $\phi,\delta$ for $A,X$ such that
 $\dom(\phi),\EQ(\delta)$ are $\bfSig^1_\alpha$-sets.
 Then the partial function $\gamma:=q \circ \delta$ is a quotient representation of $Y$.
 There are partial continuous $F,G$ on $\calN$ with $\bfPi^0_2$-domains which
 realize $f$ and $g$, respectively.
 The sets 
 \begin{align*} 
  B&:=\{ (p,s) \in \dom(\delta) \times \dom(\phi)   \,|\, (p,F(s)) \in \EQ(\delta)\}, 
 \\
  C&:=\{ (s,t) \in \dom(\phi)   \times \dom(\phi)   \,|\, (G(s),F(t)) \in \EQ(\delta) \} 
  \;\;\text{and}
  \\
   D&:=\{ (t,p) \in \dom(\phi)   \times \dom(\delta) \,|\, (G(t),p) \in \EQ(\delta)\}
 \end{align*}  
 are in $\bfSig^1_\alpha(\calN^2)$ by Proposition~\ref{p:hpro1}.
 Using the projections $\pi_i\colon \calN \to \calN$ of the inverse of the homeomorphism 
 $\langle . \rangle\colon \calN^\omega \to \calN$
 from Section~\ref{sub:topspaces}, we define for $k \geq 1$ sets $E_k,M \subseteq \calN^3$ by
 \begin{align*}
  E_k:= & \big\{ (p,r,p') \,\big|\,
        (p,\pi_1(r)) \in B \;\&\;
        \forall 1 \leq i < k.\, (\pi_i(r),\pi_{i+1}(r)) \in C 
        \;\&\; (\pi_{k}(r),p') \in D 
   \big\},  
  \\  
  M:= &  
   \big\{ (p,p',r)   \,\big|\, 
          (p,p') \in \EQ(\delta) \;\,\text{or}\;\, 
          (p,r,p') \in {\textstyle\bigcup\nolimits_{k \in \omega}} E_k \;\,\text{or}\;\,
          (p',r,p) \in {\textstyle\bigcup\nolimits_{k \in \omega}} E_k 
   \big\} \,.
 \end{align*}
 Then we have
 \[
   \EQ(\gamma)=
   \big\{ (p,p') \in \dom(\delta)^2 \,\big|\, \delta(p) \equiv \delta(p') \big\}
   =\big\{ (p,p') \in \calN^2 \,\big|\, \exists r \in \calN. (p,p',r) \in M \big\},
 \]
 because $(\delta(p),\delta(p')) \in R^{(l)}$ iff $\delta(p)=\delta(p')$ or there are 
 $k \in \{1,\dotsc,l\}$ and $r \in \calN$ such that $(p,r,p') \in E_k$ or $(p',r,p) \in E_k$.
 By Proposition~\ref{p:hpro1}, $E_k$, $M$ and $\EQ(\gamma)$ are $\bfSig^1_\alpha$-sets. 
 We conclude $Y \in \QTE(\bfSig^1_\alpha)$.
 By Proposition~\ref{p:QTE->QCBZ},
  $\mathcal{T}_0(Y) \in \QCBZ(\bfSig^1_{\alpha+2})$.
 Since $\alpha+2<\lambda$, $\QCBZ(\bfSig^1_{\alpha+2}) \subseteq \QCBZ(\bfSig^1_{<\lambda})$. Therefore, $\mathcal{T}_0(Y) \in \QCBZ(\bfSig^1_{<\lambda})$.
\end{proof}

%
%

\section{Kleene-Kreisel continuous functionals of countable types}\label{sec:functional}

Here we extend all results in \cite{scs13} about the continuous functionals of finite types
to the continuous functionals of countable types defined as follows:

\begin{definition}\label{func}
Using the function space construction of $\QCBZ$, we define the sequence of
qcb$_0$-spaces $\{\Nk{\alpha}\}_{\alpha<\omega_1}$ by induction on countable ordinals $\alpha$ as follows:
\[
  \Nk{0}:=\omega,\;
  \Nk{\alpha+1}:=\omega^{\Nk{\alpha}} \text{ and }
  \Nk{\lambda}:=\prod_{\alpha<\lambda}\Nk{\alpha}\,,
\]  
where $\omega$ denotes the space of natural numbers endowed with the discrete topology, 
$\alpha,\lambda<\omega_1$ and $\lambda$ is a limit ordinal. 
We call $\Nk{\alpha}$  \emph{the space of continuous functionals of  type $\alpha$} over $\omega$.
\end{definition}

Obviously, for $k\in \omega$ the space $\Nk{k}$  coincides with the space of
Kleene-Kreisel continuous functionals of type $k$ extensively studied in the literature \cite{no80,no81,no99}, 
and $\Nk{1}$ coincides with the Baire space $\calN$.
For any finite $k \geq 2$, the sequential topology on $\Nk{k}$ is strictly finer
than the corresponding compact-open topology \cite{hyland}.
Furthermore it is neither zero-dimensional nor regular \cite{Sch:NNN}.

Any of the introduced spaces has a natural canonical admissible representation 
$\delta_\alpha \colon D_\alpha\to\Nk{\alpha}$ induced by the constructions described 
in Proposition~\ref{p:fspace:rep} and the proof of Proposition~\ref{p:prod:coprod}. They may be defined (using the notation from Sections \ref{sub:admiss} and \ref{sec:hph:qcb}) as follows: 
\[D_0=\{n0^\omega\mid n\in \omega\},\;
  \delta_0(n0^\omega):=n,\;
  \delta_1:=id_\mathcal{N},\;
  \delta_{\alpha+1}:=\delta_0^{\delta_\alpha} \text{ and }
  \delta_\lambda:=\prod_{\alpha<\lambda}\delta_\alpha\,,
\]
where $\alpha>0$ and $\lambda$ is a limit ordinal.

The spaces of continuous functionals enjoy the following product property,
which, in the case of the finite types, belongs to the folklore of sequential spaces (\cite{EL08}).

\begin{proposition}\label{p:Na*Nb=Nb}
 For all countable ordinals $\alpha \leq \beta$, $\Nalpha \times \Nbeta$ is homeomorphic 
 to $\Nbeta$.
\end{proposition}

\begin{proof}
 We proceed by showing the following instances of our claim.
 For countable ordinals $\alpha,\beta$ and a countable limit ordinal $\lambda \geq \alpha$, 
\begin{enumerate}
\item \label{e:N*Na=Na}
 $\omega \times \Nalpha \cong \Nalpha$ 
\item \label{e:Na^N=Na}
  $(\Nalpha)^\omega \cong \Nk{\max\{1,\alpha\}}$
\item \label{e:Na*Na=Na}
  $\Nalpha \times \Nalpha \cong \Nalpha$
\item \label{e:Na+1*Nb+1=Nb+1}
  $\Nalpha \times \Nbeta \cong \Nbeta$ implies $\Nk{\alpha+1} \times \Nk{\beta+1} \cong \Nk{\beta+1}$
\item \label{e:Na*Nl=Nl}
  $\Nalpha \times \Nlambda \cong \Nlambda$ 
\item \label{e:Na*Nl+1=Nl+1}
  $\Nalpha \times \Nk{\lambda+1} \cong \Nk{\lambda+1}$  
\item \label{e:Na*Na+1=Na+1}
  $\Nalpha \times \Nk{\alpha+1} \cong \Nk{\alpha+1}$
\item \label{e:Na*Na+k=Na+k}
  $\Nalpha \times \Nk{\alpha+k} \cong \Nk{\alpha+k}$  for all $k \in \omega$.
\end{enumerate}
Now we show the above claims.
\begin{enumerate}
\item 
 For $\alpha=0$, the claim is just the statement $\omega \times \omega \cong \omega$.
 For a limit ordinal $\lambda >0$ we calculate:
 \[
   \Nlambda=\Tprod_{\beta <\lambda} \Nbeta 
   \cong \omega \times \Tprod_{1 \leq \beta <\lambda} \Nbeta
   \cong \omega \times \omega \times \Tprod_{1 \leq \beta <\lambda} \Nbeta
   \cong \omega \times \Tprod_{\beta <\lambda} \Nbeta 
   =\omega \times \Nlambda.
 \]
 Now let $\alpha$ be a successor cardinal $\alpha=\alpha'+1$.
 We fix an element $x_0 \in \Nk{\alpha'}$ and define functions
 $S\colon \Nk{\alpha} \times \omega^2 \to \Nk{\alpha}$,
 $\varphi\colon \omega \times \Nalpha \to \Nalpha$ and $\psi\colon \Nalpha \to \omega \times \Nalpha$ by
 \begin{align*}
  S(f,a,b)(x) &:= \left\{
    \begin{array}{cl}
     b & \text{if } f(x)=a             
      \\
     a & \text{if } f(x)=b
      \\
     f(x) & \text{otherwise,} 
    \end{array} \right.
  \\
  \varphi(a,f) &:=S\big( f, \langle a,f(x_0)\rangle, f(x_0) \big)
  \;\;\text{and}\;\;
  \\
   \psi(f)&:=\big( \pi_1(f(x_0)),  S( f, \pi_2(f(x_0))  ,f(x_0)) \big) 
 \end{align*}
 for $f \in \Nk{\alpha}$, $x \in \Nk{\alpha'}$ and $a,b \in \omega$,
 where $\pi_1,\pi_2:\omega \to \omega$ are the projections of the inverse of the
 pairing function $\langle .,.\rangle$.
 By cartesian closedness of $\QCBZ$, these functions are continuous.
 One easily checks that $\varphi$ and $\psi$ are inverses of each other. 
 Hence $\omega \times \Nalpha$ is homeomorphic to $\Nalpha$. 
\item
 Clearly $\Nk{0}^\omega=\Nk{1}$. For a successor ordinal $\alpha=\alpha'+1$,
 the cartesian closedness of $\QCBZ$ and Claim~\eqref{e:N*Na=Na} yield us
 \[
  \Nalpha^\omega \cong \big(\omega^{\Nk{\alpha'}}\big)^\omega
  \cong \omega^{\omega \times \Nk{\alpha'}}
  \cong \omega^{\Nk{\alpha'}}
  \cong \Nalpha.
 \]
 For limit ordinals $\lambda>0$ we proceed by induction and get by the induction hypothesis
 \begin{multline*}
  \Nlambda
  \cong \Nk{0} \times \Nk{1} \times \Tprod_{2 \leq \beta<\lambda} \Nbeta
  \cong \Tprod_{1\leq \beta< \lambda} \Nbeta
  \\
  \cong \Tprod_{1\leq \beta<\lambda} (\Nbeta)^\omega
  \cong \big(\Tprod_{1\leq \beta<\lambda} \Nbeta\big)^\omega
  \cong \Nlambda^\omega \,. \qquad
 \end{multline*}
\item
   We proceed by induction.
   For $\alpha=0$, the statement is simply $\omega \times \omega \cong \omega$.
   For any limit ordinal $\lambda \neq 0$, the induction hypothesis yields
   \[ \Nk{\lambda} \times \Nk{\lambda}
    = \big( \Tprod_{\alpha<\lambda} \Nk{\alpha}\big) \times \big(\Tprod_{\alpha<\lambda} \Nk{\alpha} \big)
    \cong \Tprod_{\alpha<\lambda} (\Nk{\alpha} \times \Nk{\alpha})
    \cong \Tprod_{\alpha<\lambda} \Nk{\alpha}
    = \Nk{\lambda}.
   \]
   For any successor ordinal $\alpha=\alpha'+1$, the cartesian closedness of $\QCBZ$ 
   and the induction hypothesis yield
   \[ \Nk{\alpha} \times \Nk{\alpha}
    =\omega^{\Nk{\alpha'}} \times \omega^{\Nk{\alpha'}}
    \cong (\omega \times \omega)^{\Nk{\alpha'}}
    \cong \omega^{\Nk{\alpha'}}
    =\Nk{\alpha}.
   \]
\item
 We calculate using the cartesian closedness of $\QCBZ$ and Claim~\eqref{e:N*Na=Na}:
 \begin{multline*}
  \Nk{\alpha+1} \times \Nk{\beta+1}
  = \omega^{\Nalpha} \times \omega^{\Nbeta}
  \cong \omega^{\Nalpha} \times \omega^{\Nalpha \times\Nbeta}
  \cong \omega^{\Nalpha} \times \big(\omega^{\Nbeta}\big)^{\Nalpha}
  \\
  \cong \big(\omega \times \omega^{\Nbeta}\big)^{\Nalpha}
  \cong \big(\omega^{\Nbeta}\big)^{\Nalpha}
  \cong \omega^{\Nalpha \times \Nbeta}
  \cong \omega^{\Nbeta}
  =\Nk{\beta+1} \,. \quad
 \end{multline*}
\item
 We calculate using Claim~\eqref{e:Na*Na=Na}:
 \begin{multline*}
  \Nalpha \times \Nlambda \cong \Nalpha \times \Tprod_{\beta<\lambda} \Nbeta
  \cong \Nalpha \times \Nalpha \times \Tprod_{\beta<\lambda, \beta \neq \alpha} \Nbeta
  \\
  \cong \Nalpha \times \Tprod_{\beta<\lambda, \beta \neq \alpha} \Nbeta
  \cong \Tprod_{\beta<\lambda} \Nbeta
  = \Nlambda \,. \quad
 \end{multline*}
\item
 For a successor ordinal $\alpha=\alpha'+1<\lambda$,
 we have $\Nk{\alpha'} \times \Nlambda \cong \Nlambda$ by Claim~\eqref{e:Na*Nl=Nl} 
 and thus $\Nalpha \times \Nk{\lambda+1} \cong \Nk{\lambda+1}$
 by Claim~\eqref{e:Na+1*Nb+1=Nb+1}.
 For a limit ordinal $\alpha$ with $0<\alpha \leq \lambda$, we proceed by ordinal induction
 and get by Claim~\eqref{e:Na^N=Na}:
 \begin{multline*}
  \Nk{\lambda+1}
  \cong (\Nk{\lambda+1})^\omega
  \cong \Tprod_{i \in \omega} \Nk{\lambda+1}
  \cong \Tprod_{\beta<\alpha} \Nk{\lambda+1}
  \cong \Tprod_{\beta<\alpha} (\Nbeta \times \Nk{\lambda+1})
  \\
  \cong \big(\Tprod_{\beta<\alpha} \Nbeta\big) \times \big(\Tprod_{\beta<\alpha}\Nk{\lambda+1} \big)  
  \cong \big(\Tprod_{\beta<\alpha} \Nbeta\big) \times (\Nk{\lambda+1})^\omega
  \cong \Nalpha \times \Nk{\lambda+1} \,.
 \end{multline*}
\item
 We proceed by ordinal induction. 
 For $\alpha=0$, $\Nalpha \times \Nk{\alpha+1} \cong \Nk{\alpha+1}$
 is an instance of Claim~\eqref{e:N*Na=Na}.
 For a successor ordinal, $\Nalpha \times \Nk{\alpha+1} \cong \Nk{\alpha+1}$
 follows from the induction hypothesis and Claim~\eqref{e:Na+1*Nb+1=Nb+1}.
 For a limit ordinal $\lambda>0$, 
 we know $\Nlambda \times \Nk{\lambda+1} \cong \Nk{\lambda+1}$
 from Claim~\eqref{e:Na*Nl+1=Nl+1}.
\item
 We proceed by induction on $k \in \omega$.
 For $k \in \{0,1\}$, we know the claim from~\eqref{e:Na*Na=Na} and~\eqref{e:Na*Na+1=Na+1}.
 For $k\geq 2$ we obtain by Claim~\eqref{e:Na*Na+1=Na+1} and by the induction hypothesis
 \begin{multline*}
   \Nalpha \times \Nk{\alpha+k} 
   \cong \Nalpha \times \Nk{\alpha+k-1} \times \Nk{\alpha+k}
   \\
   \cong \Nk{\alpha+k-1} \times \Nk{\alpha+k}
   \cong \Nk{\alpha+k} \,. \qquad
 \end{multline*}
\end{enumerate}

Now let $\alpha \leq \beta< \omega_1$. 
Then either there is some $k \in \omega$ such that $\beta=\alpha+k$
or the largest limit ordinal $\lambda$ with $\lambda \leq \beta$ satisfies $\alpha < \lambda$.
In the first case, $\Nalpha \times \Nbeta \cong \Nbeta$ follows from Claim~\eqref{e:Na*Na+k=Na+k}.
Otherwise we choose $l \in \omega$ such that $\beta=\lambda+l$ and calculate
using Claims~\eqref{e:Na*Nl=Nl} and~\eqref{e:Na*Na+k=Na+k}: 
\[
  \Nbeta
   = \Nk{\lambda+l}
   \cong \Nlambda\times \Nk{\lambda+l}
   \cong \Nalpha \times \Nlambda \times \Nk{\lambda+l}
   \cong\Nalpha \times \Nbeta \,.
\]
\end{proof}

From this propositions, we deduce the following basic properties of continuous functionals of countable types.

\begin{lemma}\label{l:N<k>:properties}
\begin{enumerate}
\item \label{e:N<k>retract}
 For all $\alpha<\beta<\omega_1$, the spaces $\Nk{\alpha}$, $\omega \times \Nk{\beta}$
 and $\calN \times \Nk{\beta}$
 are continuous retracts of $\Nk{\beta}$.
\item \label{e:Na^omega=Na}
  $(\Nk{0})^\omega\cong\Nk{1}$ and $(\Nk{\alpha})^\omega\cong\Nk{\alpha}$ for $1\leq \alpha<\omega_1$. 
\item \label{e:a<=b}
  For all $\beta<\omega_1$, $\prod_{\alpha\leq\beta}\Nk{\alpha}\cong\Nk{\beta}$.
\item \label{e:a<g<=b}
 For all $\alpha<\beta<\omega_1$, $\prod_{\alpha<\gamma\leq\beta}\Nk{\alpha}\cong\Nk{\beta}$. 
\item \label{e:b0<b1<b2}
 For all countable ordinals $\beta_0<\beta_1<\cdots$, we have
 $\prod_{k\in \omega}\Nk{\beta_k}\cong\Nk{\sup\{\beta_k\mid k\in \omega\}}$. 
\item \label{e:a0:a1:a2}
 For all countable ordinals $\alpha_0,\alpha_1,\ldots$, we have 
 $\prod_{i\in \omega}\Nk{\alpha_i}\cong\Nk{\sup\{1,\alpha_i\mid i\in \omega\}}$.
\item \label{e:Na^Nb} 
 For all $\alpha,\beta<\omega_1$, we have $\Nk{\alpha}^{\Nk{\beta}}\cong\Nk{\max\{\alpha,\beta+1\}}$.
\end{enumerate}
\end{lemma}

\begin{proof}
\begin{enumerate}
\item 
  This follows immediately from Proposition~\ref{p:Na*Nb=Nb}. 
\item 
 This has been shown inside the proof of Proposition~\ref{p:Na*Nb=Nb}. 
\item
 We proceed by induction, the case $\beta=0$ being trivial. 
 If $\beta=\gamma+1$ is successor then we have
 \[
   \Tprod_{\alpha\leq\beta}\Nk{\alpha}
   \cong\Nk{\beta}\times \Tprod_{\alpha\leq\gamma}\Nk{\alpha}
   \cong\Nk{\beta}\times \Nk{\gamma}\cong\Nk{\beta}.
 \]  
 If $\beta$ is limit then we have
 \[ \Tprod_{\alpha\leq\beta}\Nk{\alpha}
   \cong\Nk{\beta}\times \Tprod_{\alpha<\beta}\Nk{\alpha}
   \cong\Nk{\beta}\times \Nk{\beta}\cong\Nk{\beta}\,.
 \]  
\item
 By induction on $\beta$, the case $\beta=\alpha+1$ being trivial. 
 If $\beta=\gamma+1>\alpha+1$ is successor then we have
 \[
   \Tprod_{\alpha<\delta\leq\beta}\Nk{\delta}
   \cong\Nk{\beta}\times \Tprod_{\alpha<\delta\leq\gamma}\Nk{\delta}
   \cong\Nk{\beta}\times \Nk{\gamma}\cong\Nk{\beta}.
 \]  
 If $\beta$ is limit then we have by Claim~\eqref{e:a<=b}
 \begin{multline*}
  \Tprod_{\alpha<\delta\leq\beta}\Nk{\delta}
  \cong \Nk{\beta}\times \Tprod_{\alpha<\delta<\beta}\Nk{\delta}
  \cong \Nk{\beta}\times \Nk{\alpha} \times \Tprod_{\alpha<\delta<\beta}\Nk{\delta}
  \\
  \cong\Nk{\beta}\times \Tprod_{\delta\leq\alpha}\Nk{\delta}\times \Tprod_{\alpha<\delta<\beta}\Nk{\delta}
  \cong\Nk{\beta}\times \Nk{\beta}
  \cong\Nk{\beta}. \qquad
 \end{multline*}
\item
 Clearly, $\beta:=\sup\{\beta_k\mid k\in \omega\}$ is a limit ordinal.
 By Claims \eqref{e:a<=b} and \eqref{e:a<g<=b} we get
 \[
   \Nk{\beta}
   \cong \Tprod_{\alpha<\beta} \Nk{\beta}
   \cong\Tprod_{\alpha\leq\beta_0}\Nk{\alpha}\times \Tprod_{\beta_0<\alpha\leq\beta_1}\Nk{\alpha}\times \cdots 
   \cong\Nk{\beta_0}\times \Nk{\beta_1}\times\cdots=\Tprod_k\Nk{\beta_k}.
  \] 
\item
 If $\alpha_i=0$ for all $i$ then $\Tprod_i \Nk{\alpha_i}=\Nk{1}$. 
 So let $\alpha_i>0$ for some $i\in \omega$. 
 First we assume that there is some $j$ with $\alpha_j=\alpha:=\sup\{1,\alpha_i\mid i\in \omega\}$.
 From Proposition~\ref{p:Na*Nb=Nb} and Claim~\eqref{e:Na^omega=Na} we get
 \begin{multline*}
   \Tprod_{i\in \omega}\Nk{\alpha_i} 
   \cong \Tprod_{i \neq j}\Nk{\alpha_i}  \times \Nk{\alpha}
   \cong \Tprod_{i \in \omega}\Nk{\alpha_i}  \times (\Nk{\alpha})^\omega 
   \\
   \cong \Tprod_{i \in \omega}\Nk{\alpha_i} \times \Tprod_{i \in \omega}\Nk{\alpha} 
   \cong \Tprod_{i \in \omega}(\Nk{\alpha_i} \times \Nk{\alpha})
   \cong \Tprod_{i \in \omega} \Nk{\alpha} 
   \cong \Nk{\alpha} \,.
 \end{multline*}

 It remains to consider the case when $\alpha_j<\alpha$ for all $j\in \omega$ (so, in particular, $\alpha$ is limit). 
 Choose a sequence $i_0<i_1<\cdots$ of indices such that $\alpha_{i_0}<\alpha_{i_1}<\cdots$ 
 and $\sup\{\beta_k\mid k\in \omega\}=\alpha$ where $\beta_k:=\alpha_{i_k}$ for each $k\in \omega$. 
 Claims~\eqref{e:a<=b} and~\eqref{e:a<g<=b} yield us
 \[ 
  \Tprod\{\Nk{\alpha_i}\mid \alpha_i\leq\beta_0\} \cong \Nk{\beta_0} 
  \;\;\text{and}\;\;
   \Tprod\{\Nk{\alpha_i}\mid \beta_{j-1}<\alpha_i\leq\beta_j\} \cong \Nk{\beta_j},
 \] 
 for all $j \geq$. 
 From Claim \eqref{e:b0<b1<b2} we deduce
 \[
  \Nk{\alpha}
  \cong\Tprod_k\Nk{\beta_k}
  \cong \Tprod\{\Nk{\alpha_i}\mid \alpha_i\leq\beta_0\}\times \Tprod\{\Nk{\alpha_i}\mid \beta_0<\alpha_i\leq\beta_1\} \times\cdots 
  \cong\Tprod_i\Nk{\alpha_i}.
 \]
\item
 If $\alpha=0$, then $\Nk{\alpha}^{\Nk{\beta}}=\omega^{\Nk{\beta}}=\Nk{\beta+1}$.
 If $\alpha=\gamma+1$ is successor then we have by Proposition~\ref{p:Na*Nb=Nb}
 \[  
  \Nk{\alpha}^{\Nk{\beta}}
  =(\omega^{\Nk{\gamma}})^{\Nk{\beta}} 
  \cong\omega^{\Nk{\gamma}\times\Nk{\beta}}
  \cong \omega^{\Nk{\max\{\gamma,\beta\}}}=\Nk{\max\{\alpha,\beta+1\}}.
 \]  
 Now let $\alpha>0$ be a limit ordinal.
 First consider the case $\beta<\alpha$. 
 There are non-limit ordinals $\alpha_i$ such that $\beta<\alpha_0<\alpha_1<\cdots$ and $\sup\{\alpha_i\mid i\in \omega\}=\alpha$. 
 From the result for non-limit ordinals and Claim~\eqref{e:b0<b1<b2} we deduce
 \[
  \Nk{\alpha}^{\Nk{\beta}}
  \cong(\Tprod_i\Nk{\alpha_i})^{\Nk{\beta}} 
  \cong\Tprod_i(\Nk{\alpha_i}^{\Nk{\beta}}) 
  \cong\Tprod_i\Nk{\alpha_i}\cong\Nk{\alpha}.
 \] 
 Let now $\alpha\leq\beta$.
 Then there are non-limit ordinals $\alpha_i$ such that $\alpha_0<\alpha_1<\cdots$ and $\sup\{\alpha_i\mid i\in \omega\}=\alpha$.
 We get by the result for non-limit ordinals and by Claims~\eqref{e:Na^omega=Na} and~\eqref{e:b0<b1<b2}
 \[
  \Nk{\alpha}^{\Nk{\beta}}
  \cong(\Tprod_i\Nk{\alpha_i})^{\Nk{\beta}} 
  \cong\Tprod_i(\Nk{\alpha_i}^{\Nk{\beta}}) 
  \cong\Tprod_i\Nk{\beta+1}\cong\Nk{\beta+1}.
 \]  
\end{enumerate}
\end{proof}


\section{The relationship of the continuous functionals to the hyperprojective hierarchies}
\label{sec:KKCF:HP}

In this section we investigate the relationship of the continuous functionals 
to both the hyperprojective hierarchy of subsets of the Baire space $\calN$
and the hyperprojective hierarchy of qcb$_0$-spaces.

The first result generalizes the characterization of projective subsets of $\calN$
with the help of Kleene-Kreisel continuous functionals in \cite[Theorem 7.6]{scs13}
to hyperprojective subsets.

\begin{theorem}\label{charluz}
 Let $\alpha$ be a non-zero countable ordinal and $B$ a non-empty subset of $\calN$.
 Then $B\in\bfSig^1_\alpha(\calN)$ iff
 there is a continuous function $f\colon \Nk{\alpha}\to \calN$
 with $\rng(f)=B$.
\end{theorem}

The theorem is based on the following proposition about complementation in $\calN$, 
which has been shown in \cite{scs13}. 

\begin{proposition}\label{p:key}
 Let $Y$ be a qcb$_0$-space and let $f\colon Y\to\calN$ be
 a continuous function with $\rng(f) \neq \calN$.
 Then there exists a continuous function $g\colon \calN \times \omega^Y\to\calN$
 with $\rng(g)=\calN\setminus \rng(f)$.
\end{proposition}

\begin{proof*}{Theorem~\ref{charluz}}
 We proceed by induction. The case of non-limit levels are considered precisely 
 as for the finite ordinals in the proof of Proposition 7.5 and Theorem 7.6 in \cite{scs13}. 
 Let now $\lambda$ be a limit ordinal and $f\colon \Nk{\lambda}\to \calN$   a continuous function
 with $\rng(f)=B$. 
 Then $\rng(f\circ\delta_\lambda)=B$, where $\delta_\lambda:D_\lambda\to\Nk{\lambda}$ is 
 the canonical admissible representation of $\Nk{\lambda}$. 
 Since $D_\lambda$ is in $(\bfSig^1_{<\lambda})_\delta$ by Proposition \ref{p:prod:coprod}(2), 
 $B\in\bfSig^1_\lambda(\calN)$ by Proposition \ref{p:hpro1}(3).
 
For the other direction, let $B\in\bfSig^1_\lambda(\calN)$.
Then $B=g(A)$ for some $A\in(\bfSig^1_{<\lambda})_\delta$ and some continuous function 
$g\colon \calN \to \calN$. 
Since $\bfSig^1_{<\lambda}=\bfPi^1_{<\lambda}$ by Proposition~\ref{p:hpro1}, 
there are sets $A_0,A_1,\dotsc \in \bfPi^1_{<\lambda}$ 
such that $A=\bigcap_{k\in \omega}A_k$.
Then $\overline{A}=\bigcup_{k\in \omega}\overline{A}_k$, where $\overline{A}=\calN\setminus A$. 
Choose ordinals $\alpha(0),\alpha(1),\dotsc$ below $\lambda$ such that $\overline{A}_k\in\bfSig^1_{\alpha(k)}$ 
for each $k\in \omega$.
By the induction hypothesis, there are continuous functions $f_k\colon \Nk{\alpha(k)}\to \calN$ such that 
$\rng(f_k)=\overline{A}_k$ for each $k\in \omega$. 
Then there is a continuous function  $h\colon X\to \calN$ such that $\rng(h)=\overline{A}$,
where $X=\bigoplus_k \Nk{\alpha(k)}$. 
By Proposition \ref{p:key} there is a continuous function  $u\colon\calN\times\omega^X\to \calN$ 
such that $\rng(u)=A$. 
Let $\mathcal{B}$ be the set $\{1\} \cup \{\alpha(k)+1\,|\, k \in \omega\}$. 
Since $\QCBZ$ is cartesian closed and closed under countable products and countable coproducts,
we have
 \[ \calN\times\omega^X \cong 
  \calN \times \prod_{k \in \omega} \omega^{\Nk{\alpha(k)}} \cong
  \Nk{1} \times \prod_{k \in \omega} \Nk{\alpha(k)+1}.
 \]  
By Lemma~\ref{l:N<k>:properties}\eqref{e:N<k>retract} and~\eqref{e:Na^omega=Na}, 
$\Nk{1} \times \prod_{k \in \omega} \Nk{\alpha(k)+1}$
is a continuous retract of $\prod_{\beta \in \mathcal{B}} \Nk{\beta}$.
This set is in turn a continuous retract of $\prod_{\beta <\lambda} \Nk{\beta}=\Nk{\lambda}$.
So there is a continuous surjection $r$ from $\Nk{\lambda}$ onto $\calN \times \omega^X$.
Since $B=g(A)$, the continuous function $f:=g\circ u \circ r$ satisfies $B=\rng(f)$.
\end{proof*} 

Below we make use of the following lemma.

\begin{lemma}\label{fpoint}
\begin{enumerate}
 \item 
 Let $X,Y$ be qcb$_0$-spaces and $f:Y\to Y$ a continuous function without fixed points. Then there is no continuous surjection from $X$ onto $Y^X$.
 \item
 For any qcb$_0$-space $X$, there is no continuous surjection from $X$ onto $\omega^X$.
\end{enumerate}
\end{lemma}

\begin{proof}
The second assertion follows from the first one, where $Y=\omega$ and $f(y)=y+1$.
So it suffices to prove the first one. Suppose for a contradiction that $h:X\to Y^X$ is a continuous surjection. Define a function $g:X\to Y$ by $g(x)=f(h(x)(x))$. By cartesian closedness of $\QCBZ$, $g$ is continuous, i.e. $g\in Y^X$. Since $h$ is a surjection, $g=h(a)$ for some $a\in X$. Then $g(a)=f(h(a)(a))=f(g(a))$, hence $g(a)\in Y$ is a fixed point of $f$. A contradiction.
\end{proof}

Finally, we relate the continuous functionals of countable types to the hyperprojective hierarchy of
qcb$_0$-spaces (extending Theorem 7.7 of \cite{scs13}). 
The next result provides the exact estimation of the spaces of
continuous functionals of countable types in the hyperprojective hierarchy of qcb$_0$-spaces. 
On the other hand, the result provides ``natural'' witnesses for the non-collapse property of this hierarchy.

\begin{theorem}\label{th:luzrep}
 For any non-zero countable ordinal $\alpha$,
 $\Nk{\alpha+1} \in \QCBZ(\bfPi^1_\alpha)\setminus \QCBZ(\bfSig^1_\alpha)$. 
 For any countable limit ordinal $\lambda$,
 $\Nk{\lambda} \in \QCBZ((\bfPi^1_{<\lambda})_\delta)\setminus \QCBZ((\bfSig^1_{<\lambda})_\sigma)$. 
 \end{theorem}

\begin{proof}
 We proceed by induction. Obviously, $\Nk{0},\Nk{1} \in\QCBZ(\bfDelta^1_0)$.
 By Proposition \ref{p:exp}, $\Nk{2}\in \QCBZ(\bfPi^1_1)$. 
 
 Let now $\alpha \geq 2$. 
 If $\alpha$ is a successor ordinal, then $\Nk{\alpha+1}\in \QCBZ(\bfPi^1_\alpha)$ 
 by Proposition \ref{p:exp}, because $\Nk{\alpha} \in \QCBZ(\bfPi^1_\alpha) \subseteq \QCBZ(\bfSig^1_{\alpha+1})$
 by the induction hypothesis.
 If $\alpha=\lambda$ is a limit ordinal, we have
 $$\Nk{\lambda}\in \QCBZ((\bfSig^1_{<\lambda})_\delta) 
  = \QCBZ((\bfPi^1_{<\lambda})_\delta) \subseteq \QCBZ(\bfSig^1_\lambda)$$ 
 by the induction hypothesis and Propositions~\ref{p:hpro1}, \ref{p:prod:coprod}.
 Again by Proposition \ref{p:exp} we obtain $\Nk{\lambda+1}\in \QCBZ(\bfPi^1_\lambda)$. This completes the proof of the ``upper bounds''.

Now we turn to the ``lower bounds'' and first show that
$\Nk{\alpha+1} \notin \QCBZ(\bfSig^1_\alpha)$ for any non-zero countable ordinal $\alpha$.
We prove the stronger assertion that $\Nk{\alpha+1}$ has no continuous
representation $\delta$ with $D=\dom(\delta)\in\bfSig^1_\alpha(\calN)$.
Suppose for a contradiction that $\delta$ is such a continuous representation
of $\Nk{\alpha+1}$.
Then there is a continuous surjection $f$ from $\Nk{\alpha}$ onto $D$  by Theorem~\ref{charluz}. Then $\delta f$ is a continuous surjection from $\Nk{\alpha}$ onto $\Nk{\alpha+1}=\omega^{\Nk{\alpha}}$. 
This contradicts  Lemma \ref{fpoint}(2).

It remains to show that  $\Nk{\lambda} \not\in  \QCBZ((\bfSig^1_{<\lambda})_\sigma)$. Again we prove a stronger assertion that there is no continuous surjection $\delta$ from a set $D\in(\bfSig^1_{<\lambda})_\sigma$ onto $\Nk{\lambda}$.
Suppose for a contradiction that $\delta:D\to\Nk{\lambda}$ is such a surjection. Choose non-zero ordinals $\alpha_k$ and non-empty sets $D_k\in\bfSig^1_{\alpha_k}$ such that $\alpha_0<\alpha_1<\cdots$, $sup\{\alpha_k\mid k\in \omega\}=\lambda$, and $D=\bigcup_kD_k$. By Theorem~\ref{charluz}, for any $k\in \omega$ there is a continuous surjection from $\Nk{\alpha_k}$ onto $D_k$. Then there is a continuous surjection from $X:=\bigoplus_k\Nk{\alpha_k}$ onto $D$. Then $\delta f$ is a continuous surjection from $X$ onto $\Nk{\lambda}$. Since by Lemma \ref{l:N<k>:properties} $\Nk{\lambda}\cong \prod_k\Nk{\alpha_{k+1}}\cong\omega^X$, there is a continuous surjection from $X$ onto $\omega^X$. This contradicts  Lemma \ref{fpoint}(2).
\end{proof} 

%
%

\section{Categories of hyperprojective qcb$_0$-spaces}
\label{categ}

In this section we show that the hyperprojective hierarchy of qcb$_0$-spaces
gives rise to a nice cartesian closed category. This is the full
subcategory of the category $\QCBZ$ consisting of the spaces in
$\bigcup_{\alpha<\omega_1}\QCBZ(\bfSig^1_\alpha)$ as objects and all continuous
function between them as morphisms. We denote this category by
$\QCBZ(\mathbf{HP})$ and call its objects \emph{hyperprojective
qcb$_0$-spaces}.

Recall that a category is \emph{countably complete} (resp. \emph{countably co-complete}),
if it is closed under countable limits (resp. co-limits). 
We do not recall here the rather technical notions of limits and co-limits,
but remind the reader that a category is countably complete (resp. co-complete) 
iff it is closed under countable products and equalizers (resp. countable co-products and co-equalizers).
A category is cartesian closed, if it has finite products and admits for any two objects $X,Y$
an exponential $Y^X$ and an evaluation morphism $\mathit{ev}\colon Y^X\times X\to Y$ allowing curry and uncurry.

According to \cite{sch:phd}, the category $\QCBZ$ is cartesian closed, countably complete and countably co-complete. Here we discuss closure properties of some natural subcategories of $\QCBZ$ including $\QCBZ(\mathbf{HP})$ and the category $\QCBZ(\mathbf{P})$ of projective qcb$_0$-spaces from \cite{scs13}.

\begin{theorem}\label{th:QCBZ(HP):is:ccc} 
 The category $\QCBZ(\mathbf{HP})$ of hyperprojective qcb$_0$-spaces is cartesian closed, countably complete and countably co-complete.
 It inherits its exponentials, countable products, countable co-products, equalizers and co-equalizers 
 from $\QCBZ$.
\end{theorem}

\begin{proof}
By Propositions \ref{p:exp}  and \ref{p:prod:coprod}, $\QCBZ(\mathbf{HP})$ is closed under $\QCBZ$-exponentials and countable products. Since $\QCBZ(\mathbf{HP})$ is a full subcategory of the cartesian closed category $\QCBZ$, it is also cartesian closed.
The remaining properties follow from the corresponding results about countable products, countable co-products, 
equalizers and co-equalizers in Section \ref{sec:hph:qcb}.
\end{proof} 

The next result provides a characterization of $\QCBZ(\mathbf{HP})$ that avoids explicit mention 
of the hyperprojective hierarchy.
We thank Matthew de Brecht for pointing out this fact to us and for allowing us to include it into the paper.

\begin{theorem}\label{th:QCBZ(HP):is:min} 
 The category $\QCBZ(\mathbf{HP})$ is (up to homeomorphic equivalence) the smallest full subcategory of $\QCBZ$ 
 that has the closure properties from the previous theorem and contains the Sierpinski space as an object.
 So any 
 full subcategory $\sfC$ of $\QCBZ$ which contains the Sierpinski space as an object and
 inherits exponentials, countable limits and countable co-limits from $\QCBZ$ contains a homeomorphic copy 
 of any space in $\QCBZ(\mathbf{HP})$. 
\end{theorem}

\begin{proof}
 Let $\sfC$ be a full subcategory of $\QCBZ$ which contains the Sierpinski space as an object and
 inherits exponentials, countable limits and countable co-limits from $\QCBZ$. 
 So $\sfC$ is cartesian closed, countably complete and countably co-complete.
 Without loss of generality we assume that $\sfC$ is closed under homeomorphic equivalence.
 
 First we show that $\omega$ is an object of $\sfC$.
 Since $\IS$ is a $\sfC$-object and $\sfC$ is closed under countable product, $\IS^\omega$ is a $\sfC$-object.
 By being homeomorphic to $\IS^\omega$, $\Pomega$ is in $\sfC$ as well.
 Clearly, $\omega$ is homeomorphic to the subspace $M$ of $\Pomega$ consisting of all singleton subsets of $\omega$.
 We define $f,g\colon \Pomega \to \Pomega$ by
 \[
   f(p):= \left\{
   \begin{array}{cl}
     \emptyset & \text{if $p = \emptyset$}
     \\
     \{0\} & \text{otherwise}
   \end{array}\right.
  \;\;\text{and}\;\;
   g(p):= \left\{
   \begin{array}{cl}
     \{0,1\} & \text{if $p$ contains at least two numbers}
     \\
     \{0\} & \text{otherwise}
   \end{array}\right.
 \]
 Clearly, $f$ and $g$ are continuous and satisfy $f(p)=g(p) \iff p \in X$.
 Hence $M$ (together with its inclusion into $\Pomega$) is an equalizer to $f,g$ in $\QCBZ$ 
 and thus in $\sfC$. We conclude $\omega \in \sfC$.

 Next we show that any proper hyperprojective subset $D$ of the Baire space $\calN$,
 endowed with the subspace topology, is a $\sfC$-object. 
 To see this, choose a non-zero countable ordinal $\alpha$ such that 
 $\calN \setminus D\in\mathbf\Sigma^1_\alpha(\calN)$. 
 By Theorem \ref{charluz} there is a continuous surjection $f$ from $\Nk{\alpha}$ onto $\calN \setminus D$. 
 We define $F,G\colon \calN  \to \IS^{\Nk{\alpha}}$ by
 \[
   F(x)(y):=\top
   \;\;\text{and}\;\;
   G(x)(y):= \left\{
   \begin{array}{ll}
     \bot & \text{$f(y)=x$}
     \\
     \top & \text{otherwise}
   \end{array}\right.
 \]
 for all $x \in \calN$ and $y \in \Nk{\alpha}$.
 Since $\calN$ is a Hausdorff space and $\QCBZ$ is cartesian closed, $F$ and $G$ are continuous. 
 Moreover every $x \in \calN$ satisfies $F(x)=G(x)$ iff $x \notin \rng(f)$ iff $x \in D$.
 Hence $D$ is an equalizer to $F,G$ in $\QCBZ$.
 Since $\IS^{\Nk{\alpha}} \in \sfC$, we obtain $D \in \sfC$.
 As $\calN^2$ is homeomorphic to $\calN$, every subspace of $\calN^2$ with a hyperprojective carrier set
 is in $\sfC$ as well.

Finally, we employ co-equalizers to get all of the hyperprojective qcb$_0$-spaces. 
Let $X$ be a qcb$_0$-space having an admissible representation $\delta\colon D\to X$ 
such that $\EQ(\delta)$ is hyperprojective. Thus $D$ is hyperprojective and $\delta$ is a quotient map.
Above we have seen that $D$ and $\EQ(\delta)$ with the respective subspace topologies are spaces in $\sfC$.
We let $p_1,p_2\colon \EQ(\delta) \to D$ be the respective projections. 
One easily checks that $X$ and $\delta$ form a co-equalizer to $p_1,p_2$ in $\QCBZ$ and thus in~$\sfC$.
Hence $X$ is a $\sfC$-object.
\end{proof}

Next we identify a natural ``small'' cartesian closed subcategory of $\QCBZ(\mathbf{HP})$ closed under countable product.
By $\mathbf{1}$ we denote a fixed one-point space.

\begin{theorem}\label{minccc}  
 The full subcategory $\sfF:=\{ \mathbf{1}, \Nk{\alpha}\mid \alpha<\omega_1\}$ of $\QCBZ$ is cartesian closed and closed under countable $\QCBZ$-product. Moreover, $\sfF$ is the smallest such subcategory in the following sense: if $\sfC$ is a full cartesian closed subcategory  of $\QCBZ$ which is closed under countable product inherited  from $\QCBZ$ and
 contains the space $\omega$ as an object, then any $\sfF$-object is homeomorphic to a $\sfC$-object.
\end{theorem}

\begin{proof}
By Lemma \ref{l:N<k>:properties}\eqref{e:a0:a1:a2}, $\sfF$ is closed under countable $\QCBZ$-product. 
By Lemma \ref{l:N<k>:properties}\eqref{e:Na^Nb}, $\sfF$ is closed under $\QCBZ$-exponentiation, 
hence $\sfF$ is also cartesian closed.

Let now $\sfC$ be a subcategory with the specified properties.
From the proof of Proposition 8.2 in \cite{scs13} it follows that exponentials formed in $\sfC$ are homeomorphic
to the corresponding $\QCBZ$-exponentials. Since $\omega$ is a $\sfC$-object,
$\sfC$ contains  homeomorphic copies of all the spaces of continuous functionals of countable types.
\end{proof}

\begin{corollary}\label{c:QCBZ(P):is:min}
 There is no full cartesian closed subcategory $\sfC$ of $\QCBZ$
 such that $\sfC$ inherits countable products from $\QCBZ$,
 contains the discrete space $\omega$ of natural numbers
 and is contained itself in $\QCBZ(\bfSig^1_\alpha)$
 for some $\alpha < \omega_1$.
\end{corollary}

\begin{proof} 
Suppose for a contradiction that $\sfC$ were a cartesian closed subcategory
of $\QCBZ$ with the specified properties. 
By the previous theorem, there is a space $E \in \sfC$ homeomorphic to $\Nk{\alpha+1}$.
Hence $\Nk{\alpha+1}$ has an admissible representation $\delta$
such that $\mathit{EQ}(\delta) \in \bfSig^1_\alpha(\calN^2)$.
This contradicts Theorem~\ref{th:luzrep}.
\end{proof} 

We conclude this section by formulating ``finite'' versions of the results above. They are proved in the same way and provide a bit of new information to some results in \cite{scs13}. 
A category is \emph{finitely complete (resp. co-complete)} if it is closed under finite limits (resp. co-limits).

\begin{theorem}\label{fin-is:ccc}
 The category $\QCBZ(\mathbf{P})$ is the smallest full subcategory of $\QCBZ$
 that inherits exponentials, finite limits and finite co-limits from $\QCBZ$
 and contains the Sierpinski space and the space $\omega$ as  objects. 
\end{theorem}

\begin{theorem}\label{fin-minccc}
 The full subcategory $\sfF:=\{\mathbf{1},\,\Nk{k}\mid k\in \omega\}$ of $\QCBZ$ is cartesian closed and closed under finite $\QCBZ$-product. Moreover, $\sfF$ is the smallest such subcategory in the following sense: if $\sfC$ is a full cartesian closed subcategory  of $\QCBZ$ which is closed under finite product inherited  from $\QCBZ$ and
 contains the space $\omega$ as an object, then any $\sfF$-object is homeomorphic to a $\sfC$-object.
 \end{theorem}

\begin{remark}
 The above results characterize some subcategories of $\QCBZ$ in terms of minimality
 in a class of subcategories of $\QCBZ$ that enjoy certain structural properties.
 In most cases we have demanded that the corresponding constructions of new spaces are inherited from $\QCBZ$.
 However, for some of these constructions this requirement is not necessary.
 We omit the details.
\end{remark}

%
%

\section{Final Remarks}\label{sec:final}

We hope that the established closure properties of $\QCBZ(\mathbf{HP})$ 
motivates to study many other spaces of interest for Computable Analysis. 
We give two examples of sequences of function spaces 
which are interesting objects of investigation.

The first example is the sequence of hyperprojective
qcb$_0$-spaces $\{\Ra\}_{\alpha<\omega_1}$, defined by induction on $\alpha$ as follows:
 $$\Rk{0}:=\mathbb{R},\; \Rk{\alpha+1}:=\mathbb{R}^{\Rk{\alpha}}, \text{ and } 
\Rk{\lambda}:=\prod_{\alpha<\lambda}\Rk{\alpha},$$
 where $\mathbb{R}$ denotes the space of real numbers endowed with the standard Euclidean topology, 
$\alpha,\lambda<\omega_1$ and $\lambda$ is a non-zero limit ordinal. 
We call $\Ra$  \emph{the space of continuous functionals of  type $\alpha$} over $\mathbb{R}$.
Again, for finite ordinals we obtain the functionals over $\mathbb{R}$ of finite types 
which are rather popular. 
Propositions~\ref{p:exp} and \ref{p:prod:coprod} yield $\Rk{\alpha+1} \in \QCBZ(\bfPi^1_\alpha)$
with an analogous proof as for Theorem~\ref{th:luzrep}. 
With more effort one can establish $\Rk{\alpha+1} \notin \QCBZ(\bfSig^1_\alpha)$.

Recall that for any qcb$_0$-space $X$ 
by $\mathcal{O}(X)$ we denote the hyperspace of open sets in $X$ endowed with the Scott topology.
The space $\mathcal{O}(X)$ is  homeomorphic to the function space $\mathbb{S}^X$,
where $\mathbb{S}$ is  the Sierpinski space.
We define the sequence of  hyperprojective qcb$_0$-spaces 
$\{\mathcal{O}^\alpha(X)\}_{\alpha<\omega_1}$ by induction on $\alpha$ as follows:
$\mathcal{O}^0(X):=X$, $\mathcal{O}^{\alpha+1}(X):=\mathcal{O}(\mathcal{O}^\alpha(X))$, 
and $\mathcal{O}^\lambda(X):=\prod_{\alpha<\lambda}\mathcal{O}^\alpha(X)$, 
where again $\alpha,\lambda<\omega_1$ and $\lambda \neq 0$ is a limit ordinal. 
It seems to be worth investigating
the behavior of  $\{\Ra\}_{\alpha<\omega_1}$ 
and $\{\mathcal{O}^\alpha(X)\}_{\alpha<\omega_1}$ with respect to the hyperprojective hierarchy 
of qcb$_0$-spaces.



\begin{thebibliography}{asaa67}

\bibitem[Br13] {br} 
M.~de Brecht.
Quasi-Polish spaces. 
\emph{Annals of pure and applied logic}, 164 (2013), 356--381.

\bibitem[En89]{en89} 
R. Engelking. 
\emph{General Topology}.
Heldermann, Berlin, 1989.

\bibitem[EL08]{EL08}
 M.~Escard\'o, J.~Lawson. Personal communication.

\bibitem[ELS04]{ELS04}
 M.~Escard\'o, J.~Lawson, A.~Simpson.
 Comparing Cartesian closed Categories of Core Compactly Generated Spaces.
 \emph{Topology and its Applications} 143 (2004), 105--145.

\bibitem[GH80] {g03} 
G. Giertz, K.H. Hoffmann, K. Keimel, J.D. Lawson, M.W. Mislove, D.S. Scott. 
\emph{A compendium of Continuous Lattices}, 
Berlin, Springer, 1980.

\bibitem [Hy79]{hyland}
J.M.L. Hyland. 
\newblock Filter spaces and continuous functionals.
\newblock \emph{Annals of Mathematical Logic}, 16 (1979), 101--143.

\bibitem[Ke83]{ke83} 
A.S. Kechris. 
Suslin cardinals, k-Suslin sets and the scale property in the hyperprojective hierarchy. 
The Cabal Seminar, v. 1: Games, Scales and Suslin Cardinals, 
Eds. A.S. Kechris, B. L\"owe, J.R. Steel, Lecture Notes in Logic, 31, 2008, p. 314--332. 
(Reprinted from Lecture Notes in mathematica, No 1019, Berlin, Springer, 1983). 

\bibitem[Ke95]{ke95} 
A.S. Kechris. 
\emph{Classical Descriptive Set Theory.}
Springer, New York, 1995.

\bibitem[Kl59]{kl59} S.C. Kleene. 
Countable functionals. 
\emph{Constructivity in Mathematics} (A. Heyting, Ed.), 
North Holland, Amsterdam, 1959, 87--100.

\bibitem[Kr59]{kr59} G. Kreisel. 
Interpretation of analysis by means of constructive functionals of finite types. 
\emph{Constructivity in Mathematics} (A. Heyting, Ed.), North Holland, Amsterdam, 1959, 101--128.

\bibitem[KW85] {kw85} 
C. Kreitz and K. Weihrauch. 
Theory of representations. 
\emph{Theoretical Computer Science}, 38 (1985), 35--53.

\bibitem [No80]{no80}
D.~Normann.
\newblock Recursion on the Countable Functionals.
\newblock \emph{Lecture Notes in Mathematics},  811 (1980), Springer, Heidelberg.

\bibitem [No81]{no81}
D. Normann.
\newblock Countable functionals and the projective hierarchy.
\newblock \emph{Journal of Symbolic Logic}, 46.2 (1981), 209--215.

\bibitem [No99] {no99}
D. Normann.
\newblock The continuous functionals.
\newblock \emph{Handbook of Computability Theory} (E.R. Griffor Ed.), 1999, 251--275, Elsevier, Amsterdam.

\bibitem[Sch02] {sch:ext} 
M. Schr\"oder. 
Extended admissibility. 
\emph{Theoretical Computer Science}, 284 (2002), 519--538.

\bibitem[Sch03] {sch:phd} M. Schr\"oder. 
\emph{Admissible representations for continuous computations}. 
PhD thesis, Fachbereich Informatik, FernUniversit\"{a}t Hagen, 2003.

\bibitem[Sch09]{Sch:NNN}
M.~Schr\"oder.
The sequential topology on $\mathbb{N}^{\mathbb{N}^{\mathbb{N}}}$ is not regular.
\emph{Mathematical Structures in Computer Science}, 19
(2009), 943--957.

\bibitem[ScS13] {scs13} 
M.~Schr\"oder and V. Selivanov. 
Some Hierarchies of qcb$_0$-Spaces. 
\emph{Mathematical Structures in Computer Science}, to appear.  	
arXiv:1304.1647  

\bibitem[ScS14]{scs14} 
M.~Schr\"oder and V. Selivanov. 
Hyperprojective Hierarchy of qcb$_0$-space.
In: Computability in Europe: Language, Life, Limits.
\emph{Lecture Notes in Computer Science}, to appear.


\bibitem[Se06] {s06} V.L. Selivanov. 
Towards a descriptive set theory for domain-like structures. 
\emph{Theoretical Computer Science}, 365 (2006), 258--282.

\bibitem[Se13]{s13}
V.L. Selivanov. Total representations. 
\emph{Logical Methods in Computer Science}, 9(2) (2013), p. 1--30.
DOI: 10.2168/LMCS-9(2:5)2013.

\bibitem[We00] {wei00} 
K. Weihrauch. 
\emph{Computable Analysis.} 
Berlin, Springer, 2000.


\end{thebibliography}
\end{document}